\documentclass[onecolumn,11pt,letter]{article}
\usepackage{mathptmx}
\usepackage{epsfig}
\usepackage{epstopdf}
\usepackage{graphicx}
\usepackage{amsfonts}
\usepackage{amsmath}
\usepackage{amsthm}
\usepackage{proof}
\usepackage{latexsym}
\usepackage{amssymb}
\usepackage{color}
\usepackage{comment}
\usepackage[ruled]{algorithm2e}
\usepackage[margin=1in]{geometry}

\DeclareMathAlphabet{\mathcal}{OMS}{cmsy}{m}{n}
\usepackage{tweaklist}
\renewcommand{\enumhook}{\setlength{\itemsep}{-0.5ex}}
\renewcommand{\itemhook}{\setlength{\itemsep}{-0.5ex}}

\usepackage{graphicx}
\usepackage{amsfonts}
\usepackage{amssymb}

\usepackage{amsmath}

\begin{document}
%

\newtheorem{conj}{Conjecture}
\newtheorem{lemma}{Lemma}
\newtheorem{claim}{Claim}
\newtheorem{definition}{Definition}
\newtheorem{corol}{Corollary}
\newtheorem{theorem}{Theorem}
\newtheorem{observation}[theorem]{Observation}
\newtheorem{algo}{Algorithm}
\newtheorem{corollary}{Corollary}

\newcommand{\akc}{\textsc{Anchored $k$-Core}\xspace}
\newcommand{\AKC}{\textsc{AKC}\xspace}
\newcommand{\sat}{\textsc{Restricted-Planar-3-SAT}\xspace}

\newcommand*\samethanks[1][\value{footnote}]{\footnotemark[#1]}



\title{Preventing Unraveling in Social Networks Gets Harder\thanks{Supported by the European Research
Council (ERC) via grant Rigorous Theory of Preprocessing, reference 267959}}
\author{Rajesh Chitnis\thanks{Dept.\ of Computer Science , University of Maryland at College Park, USA. Supported in part by
NSF CAREER award 1053605, NSF grant CCF-1161626, ONR YIP award N000141110662, DARPA/AFOSR grant FA9550-12-1-0423, a University
of Maryland Research and Scholarship Award (RASA) and a Summer International Research Fellowship from the University of
Maryland. Email: \tt{rchitnis@cs.umd.edu}} \and Fedor V. Fomin\thanks{Department of Informatics, University of Bergen, PB
7803, 5020 Bergen, Norway. Email: \tt{\{fedor.fomin, petr.golovach\}@ii.uib.no}} \and Petr A. Golovach\samethanks}

\maketitle

\begin{abstract}
The behavior of users in social networks is often observed to be affected by the actions of their friends. Bhawalkar et
al.~\cite{bhawalkar-icalp} introduced a formal mathematical model for user engagement in social networks where each individual
derives a benefit proportional to the number of its friends which are engaged. Given a threshold degree $k$ the equilibrium
for this model is a maximal subgraph whose minimum degree is $\geq k$. However the dropping out of individuals with degrees
less than $k$ might lead to a cascading effect of iterated withdrawals such that the size of equilibrium subgraph becomes very
small. To overcome this some special vertices called ``anchors" are introduced: these vertices need not have large degree.
Bhawalkar et al.~\cite{bhawalkar-icalp} considered the \textsc{Anchored $k$-Core} problem: Given a graph $G$ and integers $b,
k$ and $p$ do there exist a set of vertices $B\subseteq H\subseteq V(G)$ such that $|B|\leq b, |H|\geq p$ and every vertex
$v\in H\setminus B$ has degree at least $k$ is the induced subgraph $G[H]$. They showed that the problem is NP-hard for $k\geq
2$ and gave some inapproximability and fixed-parameter intractability results. In this paper we give improved hardness results
for this problem. In particular we show that the \textsc{Anchored $k$-Core} problem is W[1]-hard parameterized by $p$, even
for $k=3$. This improves the result of Bhawalkar et al.~\cite{bhawalkar-icalp} (who show W[2]-hardness parameterized by $b$)
as our parameter is always bigger since $p\geq b$. Then we answer a question of Bhawalkar et al.~\cite{bhawalkar-icalp} by
showing that the \textsc{Anchored $k$-Core} problem remains NP-hard on planar graphs for all $k\geq 3$, even if the maximum
degree of the graph is $k+2$. Finally we show that the problem is FPT on planar graphs parameterized by $b$ for all $k\geq 7$.
\end{abstract}

\section{Introduction}
\label{sec:intro}


A social network can be thought as  the graph of relationships and interactions within a group of individuals. Social networks
play a leading role in various fields such as social sciences~\cite{social-1,social-2}, life sciences~\cite{life-1,life-2} and
medicine~\cite{life-1,medicine}. Social networks today perform a fundamental role as a medium for the spread of information,
ideas, and influence among its members. As an example, Facebook reported a figure of one billion active users as of October
2012~\cite{facebook-one-billion}. An important characteristic of social networks is that the behavior of an individual is
often influenced by the actions of their friends. New events occur quite often in social networks: some examples are usage of
a particular cell phone brand, adoption of a new drug within the medical profession, or the rise of a political movement in an
unstable society. To estimate whether these events or ideas spread extensively or die out soon, we need to model and study the
dynamics of \emph{influence propagation} in social networks. We consider the following model of \emph{user engagement} defined
by Bhawalkar et al.~\cite{bhawalkar-icalp}: there is a single product and each individual has two options of ``engaged" or
``drop out". Initially we assume that all individuals are engaged. There is a given threshold parameter $k$ such that a person
finds it worthwhile to remain engaged if and only if at least $k$ of her friends are still engaged. For example engagement
could represent active participation in a social network, and individuals might drop out and switch to a new social network if
less than $k$ of his friends are active on the current social network. Indeed such a phase transition has been observed in the
popularity of social networks: in India the leading social network was Orkut until Facebook surpassed it in August
2010~\cite{facebook-beats-orkut}.

In our model of \emph{user engagement} all individuals with less than $k$ friends will clearly drop out. Unfortunately this
can be contagious and may affect even those individuals who initially had more than $k$ friends in the social network. An
extreme example of this is given in page 17 of ~\cite{schelling2006micromotives}: consider a path on $n$ vertices and let
$k=2$. Note that $n-2$ vertices have degree two in the network. However there will be a \emph{cascade of iterated
withdrawals}. An endpoint has degree one, it drops out and now its neighbor in the path has only one friend in the social
network and it drops out as well. It is not hard to see that the whole network eventually drops out. In general at the end of
all the iterated withdrawals the remaining engaged individuals form a unique maximal induced subgraph whose minimum degree is
at least $k$. This is called as the \emph{$k$-core} and is a well-known concept in the theory of social networks. It was
introduced by Seidman~\cite{seidman-k-core} and also been studied in various social sciences
literature~\cite{chwe1999structure,chwe2000communication}.

\textbf{A Game-Theoretic Model:} Consider the following game-theoretical model~\cite{bhawalkar-icalp}: each user in a social
network pays a cost of $k$ to remain engaged. On the other hand it receives a profit of $1$ from every neighbor who is
engaged. The ``network effects" come into play, and an individual decides to remain engaged if has non-negative payoff, i.e.,
it has at least $k$ neighbors who are engaged. The $k$-core can be viewed as the unique maximal equilibrium in this model.
Assuming that all the players make decisions simultaneously the model can be viewed as a simultaneous-move game where each
individual has two strategies viz. remaining engaged or dropping out. Consider a graph $G$ defined on the set of players by
adding an edge between two players if and only if they are friends in the network. For every strategy profile $\delta$ let
$S_{\delta}$ denote the set of players
who remain engaged. The payoff for a person $v$ is 0 if she is not engaged, otherwise it is the number of her friends among engaged players minus $k$. 
We can easily characterize the set of pure Nash equilibria for this game: a strategy profile $\delta$ is a Nash equilibrium if
and only if the following two conditions hold:
\begin{itemize}
\item No engaged player wants to drop out, i.e., minimum degree of the induced graph $G[S_{\delta}]$ is $\geq k$
\item No player who has dropped out wants to become engaged, i.e., no $v\in V(G)\setminus S_{\delta}$ has $\geq k$
    neighbors in $S_{\delta}$
\end{itemize}
In general there can be many Nash equilibria. For example, if $G$ itself has minimum degree $\geq k$ then
$S_{\delta}=\emptyset$ and $S_{\delta}=V(G)$ are two equilibria (and there may be more). Recall that the goal of the product
company is to attract as many people as possible. Owing to the fact that it is a maximal equilibrium, the $k$-core has the
special property that it is beneficial to both parties: it maximizes the payoff of every user, while also maximizing the
payoff of the product company. Chwe~\cite{chwe1999structure,chwe2000communication} and
S{\"a}{\"a}skilahti~\cite{saaskilahti2007monopoly} claim that one can reasonably expect this maximal equilibrium even in
real-life implementations of this game.\\

\textbf{Preventing Unraveling:} The unraveling described above in Schelling's example of a path is highly undesirable since
the goal is to keep as many people engaged as possible. How can we attempt to prevent this unraveling? In Schelling's example
it is easy to see: if we ``buy" the two end-point players into being engaged then the whole path becomes engaged.

In general we overcome the issue of unraveling by allowing some ``anchors": these are vertices that remain engaged
irrespective of their payoffs. This can be achieved by giving them extra incentives or discounts. The hope is that with a few
\emph{anchors} we can now ensure a large subgraph remains engaged. This subgraph is now called as the \emph{anchored
$k$-core}: each non-anchor vertex in this induced subgraph must have degree at least $k$ while the anchored vertices can have
arbitrary degrees. We use the notation $\text{deg}_{S}(v)$ to denote the degree of $v$ in the graph $S$. Bhawalkar et
al.~\cite{bhawalkar-icalp} formally defined the \textsc{Anchored $k$-Core} problem :

\begin{center}
\noindent\framebox{\begin{minipage}{4.00in}\textbf{The \textsc{Anchored $k$-Core} Problem (AKC)}\\
\emph{Input }: An undirected graph $G=(V,E)$ and integers $b, k$ \\
\emph{Question}: Find a
set of vertices $H\subseteq V$ of maximum size such that
                \begin{itemize}
                \item There is a set $B\subseteq H$ and $|B|\leq b$
                \item Every $v\in H\setminus B$ satisfies $\text{deg}_{G[H]}(v)\geq k$
                \end{itemize}
\end{minipage}}
\end{center}

The AKC problem deals with finding a small group of individuals whose engagement is essential for the health of the social
network. We call the set $B$ as \emph{anchors}, the set $H$ as the \emph{anchored $k$-core} and the set $H\setminus B$ as the
\emph{anchored $k$-supercore}. The decision version of the \textsc{Anchored $k$-Core} problem deals with anchoring a given
number of vertices to maximize the number of engaged vertices. More formally:
%

\begin{center}
\noindent\framebox{\begin{minipage}{4.00in}\textbf{$p$-AKC}\\
\emph{Input }: An undirected graph $G=(V,E)$ and integers $b, k, p$ \\
\emph{Question}: Do there exist
sets $B\subseteq H\subseteq V$ such that
                \begin{itemize}
		\item $|B|\leq b$ and $|H|\geq p$
                \item Every $v\in H\setminus B$ satisfies $\text{deg}_{G[H]}(v)\geq k$
                \end{itemize}
\end{minipage}}
\end{center}


\textbf{Previous Work:} Bhawalkar et al.~\cite{bhawalkar-icalp} introduced the \textsc{Anchored $k$-Core} problem and gave
some positive and negative results for this problem. Noting that the problem is trivial for $k=1$, they showed that AKC is
polynomial time solvable for $k=2$ but NP-hard for all $k\geq 3$. They also gave a strong inapproximability result: it is
NP-hard to approximate the AKC problem to within an $O(n^{1-\epsilon})$ factor for any $\epsilon>0$. From the viewpoint of
parameterized complexity they showed that for every $k\geq 3$ the $p$-AKC problem is W[2]-hard with respect to $b$. Finally on
the positive side they give a polynomial time algorithm on graphs of bounded treewidth. On graphs with treewidth at most $w$
their algorithm
runs in $O(3^{w}(k+1)^{2w}b^{2})\cdot \text{poly}(n)$ time.\\


\textbf{Our Results:}
It is easy to see that \textsc{Anchored $k$-Core} can be solved in time $n^{b+O(1)}$, as we can try all subsets $B$ of size
$b$ of the set of vertices of the input graph, and for each $B$, find the unique $k$-core $H$ of maximum size such that
$\text{deg}_{G[H]}(v)\geq k$ if $v\in H\setminus B$ by the consecutive deletions of small degree vertices. We show that this
result is optimal in some sense by proving that $p$-AKC problem is W[1]-hard parameterized by $b+k+p$. We also show that the
$p$-AKC problem is W[1]-hard parameterized by $p$ even for $k=3$. This improves the result of Bhawalkar et
al.~\cite{bhawalkar-icalp} (who show W[2]-hardness parameterized by $b$), because our parameter is always bigger since $p\geq
b$. Bhawalkar et al. raised the question of resolving the complexity of the \textsc{Anchored $k$-Core} problem on special
graph classes. In this paper we consider the complexity of the AKC problem on the class of planar graphs. We show that the
\textsc{Anchored $k$-Core} problem is NP-hard on planar graphs for all $k\geq 3$, even if the maximum degree of the graph is
$k+2$. Finally on the positive side we show that the $p$-AKC problem on planar graphs is FPT parameterized by $b$ for all
$k\geq 7$.

\section{Fixed-Parameter Intractability Results}

In this section we give two
parameterized intractability results. Before that we give a brief introduction to parameterized complexity.\\

\textbf{Parameterized Complexity:} \emph{Parameterized Complexity} is basically a two-dimensional generalization of ``P vs.
NP'' where in addition to the overall input size $n$, one studies the effects on computational complexity of a secondary
measurement that captures additional relevant information. This additional information can be, for example, a structural
restriction on the input distribution considered, such as a bound on the treewidth of an input graph or the size of solution
set. For general background on the theory see~\cite{downey-fellows-book,flum-grohe-book,niedermeier-book}. For decision
problems with input size $n$, and a parameter $k$, the two dimensional analogue (or generalization) of P, is solvability
within a time bound of $O(f(k)n^{O(1)})$, where $f$ is a computable function of $k$ alone.
Problems having such an algorithm are said to be \emph{fixed parameter tractable} (FPT). Such algorithms are practical when
small parameters cover practical ranges. The $W$-hierarchy is a collection of computational complexity classes: we omit the
technical definitions here. The following relation is known amongst the classes in the $W$-hierarchy: $FPT=W[0]\subseteq
W[1]\subseteq W[2]\subseteq \ldots$. It is widely believed that $FPT\neq W[1]$, and hence if a problem is hard for the class
$W[i]$ (for any $i\geq 1$) then it is considered to be fixed-parameter intractable.\\



\textbf{W[1]-hardness parameterized by b+k+p:} In this section we show that the $p$-AKC problem is W[1]-hard even when
parameterized by $b+k+p$. We reduce from the well-known W[1]-hard problem \textsc{Clique}


\begin{center}
\noindent\framebox{\begin{minipage}{3.5in}\textbf{\textsc{Clique}}\\
\emph{Input }: An undirected graph $G=(V,E)$ and an integer $\ell$ \\
\emph{Question}: Does $G$ have a clique of size at least $\ell$ ?
\end{minipage}}
\end{center}

\begin{theorem}
The $p$-AKC problem is W[1]-hard parameterized by $b+k+p$ for $k\geq 3$.
\end{theorem}
\begin{proof}
Consider an instance $G=(V,E)$ of \textsc{Clique}. Let $V=\{v_1,v_2,\ldots,v_n\}$. For each edge $e=\{v_i,v_j\}$ we subdivide
it and add a new vertex $w_{\{i,j\}}$. Let this new graph be $G'=(V',E')$. Define $b=\binom{\ell}{2},\ k=\ell-1$ and
$p=\ell+\binom{\ell}{2}$. The claim is the instance $(G',b,k,p)$ of $p$-AKC answers YES and only if the instance $(G,\ell)$ of
\textsc{Clique} answers YES.

Suppose $G$ has a clique of size $\ell$ say $C=\{v_1,v_2,\ldots,v_{\ell}\}$. Pick as anchors in $G'$ the $\binom{\ell}{2}$
vertices of the type $w_{\{i,j\}}$ for $1\leq i\neq j\leq \ell$. It is easy to see that these anchors along with the vertex
set $C$ form a $k$-core of size $\ell+\binom{\ell}{2}$ for the instance $(G',b,k,p)$ of $p$-AKC.

Suppose that the instance $(G',b,k,p)$ of $p$-AKC has a solution. Since any newly added vertex has degree two we know that
every vertex in the \emph{$k$-supercore} must be from $V$. Further each vertex from the supercore needs at least $k$ anchors,
and each anchor can be shared between at most two vertices from $k$-supercore. Let $S$ denote the $k$-supercore. Counting the
number of anchors (with repetitions) in two ways we have $2\binom{\ell}{2} \geq |S|(\ell-1)$, i.e., $\ell \geq |S|$. But we
know that $|S|\geq p-b\geq \ell$, and hence $|S|=\ell$. Without loss of generality let the vertices of $S$ be
$v_1,v_2,\ldots,v_{\ell}$. Each $v_i$ has degree at least $k=\ell-1$ in $H=S\cup B$, and so $S$ must form a clique in $G$.
\end{proof}

\section{W[1]-hardness parameterized by $p$}

Bhawalkar et al. (Theorem 3 in~\cite{bhawalkar-icalp}) showed that the $p$-AKC problem is W[2]-hard parameterized by $b$ for
every $k\geq 3$. In this section we prove that it is in fact W[1]-hard parameterized by $p$ for $k=3$.
This improves on the
result of Bhawalkar et al.~\cite{bhawalkar-icalp} in two aspects: firstly our parameter is stronger since $p\geq b$ by
definition. Also we show hardness for a smaller complexity class since it is known that $FPT\subseteq W[1]\subseteq W[2]$.

\begin{figure}[h]
 \centering
 \includegraphics[height=4in]{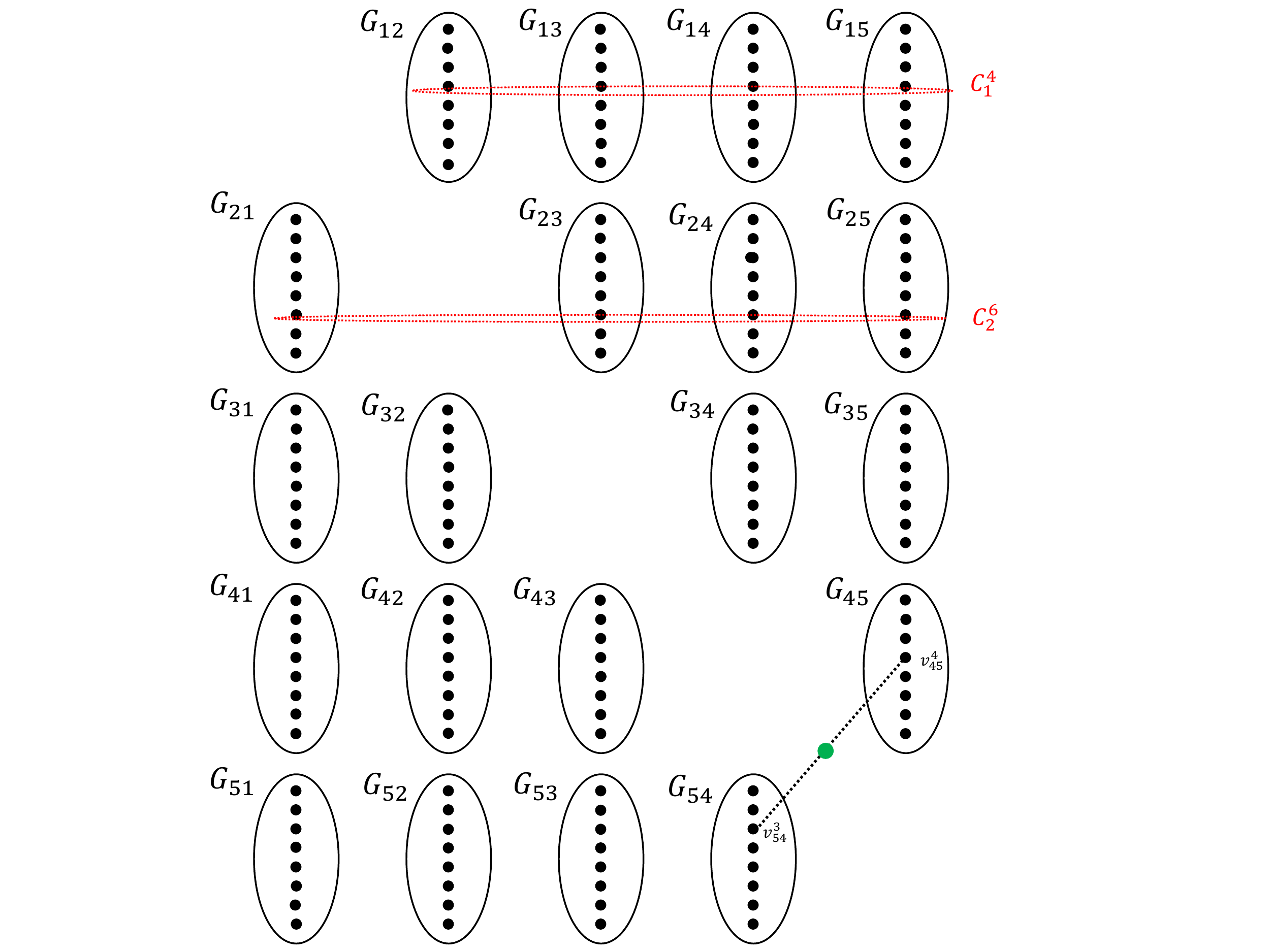}
 \caption{The graph $G'$ constructed in Theorem~\ref{thm:w[1]-hardness-main} for the special case when $n=8$ and $\ell=5$.
\label{fig:w[1]-hardness}}
\end{figure}

\begin{theorem}
The $p$-AKC problem is W[1]-hard parameterized by $p$ for $k=3$. \label{thm:w[1]-hardness-main}
\end{theorem}
\begin{proof}
We again reduce from the \textsc{Clique} problem. Consider an instance $(G=(V,E), \ell)$ of \textsc{Clique} where
$V=(v^1,v^2,\ldots,v^n)$. Construct a new graph $G'=(V',E')$ as follows. For each $1\leq i\neq j\leq \ell$ make a copy
$G_{ij}$ of the vertex set $V$ (do not add any edges). Let the vertex $v^r$ in the copy $G_{ij}$ be labeled $v^{r}_{ij}$. Add
the following edges to $G'$:
\begin{itemize}
\item For each $1\leq i\neq j\leq \ell$ and $r,s\in [\ell]$ we add an edge between $v^{r}_{ij}$ and $v^{s}_{ji}$ if and
    only if $v^{r}v^{s}\in E$. Subdivide each such edge by adding a new vertex called \emph{green}.
\item For each $1\leq i, j\leq \ell$ add the cycle $v^{j}_{i1}-v^{j}_{i2}-\ldots
    v^{j}_{i,i-1}-v^{j}_{i,i+1}-\ldots-v^{j}_{i\ell}-v^{j}_{i1}$. Let us denote this cycle by $C^{j}_{i}$.
\end{itemize}
This completes the construction of $G'$. Let $k=3,\ b=\binom{\ell}{2}$ and $p=3b$. The claim is that the instance $(G,\ell)$
of \textsc{Clique} answers YES if and only if the instance $(G',b,k,p)$ of $p$-AKC answers YES.

Suppose $G$ has a clique of size $\ell$, say $C=\{v_1,v_2,\ldots,v_{\ell}\}$. For $1\leq i\neq j\leq \ell$ pick the vertex
$v^{i}_{ij}$ from $G_{ij}$. This gives a set $S$ of $2b$ vertices. It is easy to see that $G'[S]$ consists of disjoint cycles,
and hence is regular of degree two. Now for every $v^{i}_{ij}$ there is an unique vertex $v^{j}_{ji}$ which is connected to it
by a subdivided edge in $G'$. The green vertices of this subdivided edges become the anchors. Note that we pick exactly
$\frac{|S|}{2} = b$ anchors. It is easy to see that each vertex of $S$ now has degree three in the resulting induced subgraph,
and hence $S$ becomes the $k$-supercore. Therefore the instance $(G',b,k,p)$ of $p$-AKC answers YES.

Now suppose that the instance $(G',b,k,p)$ of $p$-AKC answers YES. Let $S$ be the $k$-supercore and $B$ be the set of anchors.
Then we know that $|S|\geq p-b =2b$. Each green vertex has degree two in $G'$, and hence cannot be in $S$. Each vertex in $S$
needs at least one green vertex to achieve degree three in $G'[S\cup B]$. But any green vertex can be used by at most two
vertices in $S$. Therefore we have $2b\geq 2|B|\geq |S|\geq 2b$ which implies $|B|=b$ and $|S|=2b$. Hence the budget must come
from the green vertices only, and that the vertices of $S$ form a \emph{matching} under the relation of sharing a common green
vertex. Without loss of generality let $v^{i_1}_{12}$ and $v^{i_2}_{21}$ be two vertices in $S$ such that they share a green
vertex from $B$. Now we know that $v^{i_1}_{12}$ has degree at least three in $G'[B\cup S]$ but cannot be incident to any
other green vertex. So we need to include $v^{i_1}_{13}$ and $v^{i_1}_{1\ell}$ in $S$. Again each of these two vertices can be
incident to at most one green vertex in $G'[B\cup S]$ and ultimately this means that we must have $C^{i_1}_{1}\subseteq
G'[B\cup S]$. For $2\leq j\leq \ell$ we know that the vertex $v^{i_1}_{1j}$ needs one more edge to achieve degree at least
three. This edge must be towards a green vertex which is adjacent to some vertex in $G_{j1}$, say $v^{i_j}_{j1}$. By reasoning
similar to above we must have $C^{i_j}_{j}\subseteq G'[B\cup S]$ for every $2\leq j\leq \ell$. So we have chosen $2b$ vertices
in $S$, which is the maximum allowed budget. Therefore $S\cap G_{jj'} = \{v^{i_j}_{jj'}\}$ for every $1\leq j\neq j'\leq
\ell$.

The claim is that the set $\{v^{i_1},v^{i_2},\ldots,v^{i_{\ell}}\}$ forms a clique in $G$. Consider any two indices $1\leq
q\neq r\leq \ell$. We know that the vertex $v^{i_q}_{qr}$ is in $S$ and has degree two in $G'[B\cup S]$ as it is in the cycle
$C^{i_q}_{q}$. To achieve degree three it must be incident to some green vertex. Also it must share this green vertex with
some other vertex from $G_{rq}\cap S$. But we know that $G_{rq}\cap S= \{v^{i_r}_{rq}\}$. Therefore $v^{i_q}_{qr}$ and
$v^{i_r}_{rq}$ share a green vertex, i.e., $v^{i_q}$ and $v^{i_r}$ are adjacent in $G$, i.e, the vertices
$\{v^{i_1},v^{i_2},\ldots,v^{i_{\ell}}\}$ form a clique in $G$.
\end{proof}

\section{NP-hardness Results on Planar Graphs}

Bhawalkar et al.~\cite{bhawalkar-icalp} raised the question of investigating the complexity of the \textsc{Anchored $k$-Core}
problem on special cases of graphs such as planar graphs. In this section we provide some answers by showing NP-hardness
results for planar graphs for $k\geq 3$. The case $k\geq 4$ can be handled by a single reduction, but $k=3$ is more
complicated and requires a separate reduction. We reduce from the following problem which was shown to be NP-hard by Dahlhaus
et al.~\cite{planar-3-sat}:

\begin{center}
\noindent\framebox{\begin{minipage}{5in}\textbf{\textsc{Restricted-Planar-3-SAT}}\\
\emph{Input }: A Boolean CNF formula $\phi$ such that
                \begin{itemize}
                \item Each clause has at most $3$ literals
                \item Each variable is used in at most $3$ clauses
                \item Each variable is used at least once in positive and at least once in negation
                \item The graph $G_{\phi}$ (described below) is planar
                \end{itemize}
\emph{Question}: Is the formula $\phi$ satisfiable ?
\end{minipage}}
\end{center}

Consider an instance $\phi$ of \textsc{Restricted-Planar-3-SAT} with variables $x_1, x_2,\ldots, x_n$ and clauses
$C_1,C_2,\ldots,C_m$. We associate the following graph $G_{\phi}$ with $\phi$:
\begin{itemize}
\item For each $1\leq i\leq n$ introduce the vertices $r_i, x_i$ and $\overline{x_i}$. Add the edges $r_{i}x_{i}$ and
    $r_{i}\overline{x_i}$.
\item For each $1\leq j\leq m$ introduce the vertex $c_j$.
\item For each $1\leq i\leq n$ and $1\leq j\leq m$ add an edge between $x_i$ (or $\overline{x_i}$) and $c_j$ iff $x_i$ (or
    $\overline{x_i}$) belongs to the clause $C_j$.
\end{itemize}

\subsection{NP-hardness on Planar Graphs for $k=3$}

In this section we show that the \textsc{Anchored $k$-Core} problem is NP-hard on planar graphs for all $k=3$, even in graphs
of maximum degree $5$.

\begin{figure}[h]
 \centering
 \includegraphics[height=3.25in]{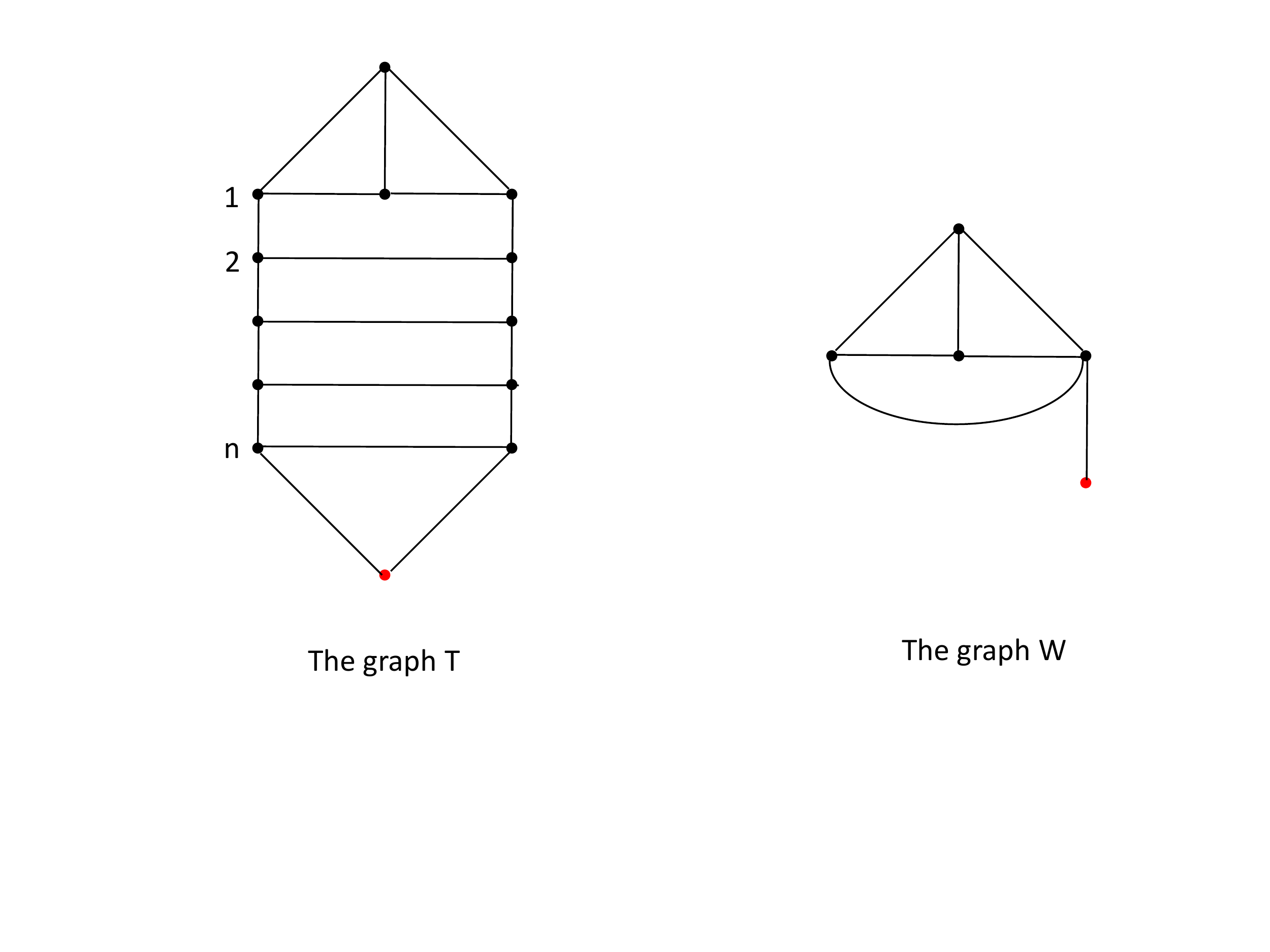}
 \vspace{-25mm}
 \caption{The graphs $T$ and $W$ used in construction of $G$ in Theorem~\ref{thm:planar-k-3}. Note that $T$ has $2n+3$ vertices, and exactly one vertex
 has degree two. The graph $W$ has exactly one vertex
of degree one.
 \label{fig:t-and-w}}
\end{figure}

\begin{figure}[h]
 \centering
 \includegraphics[height=5in]{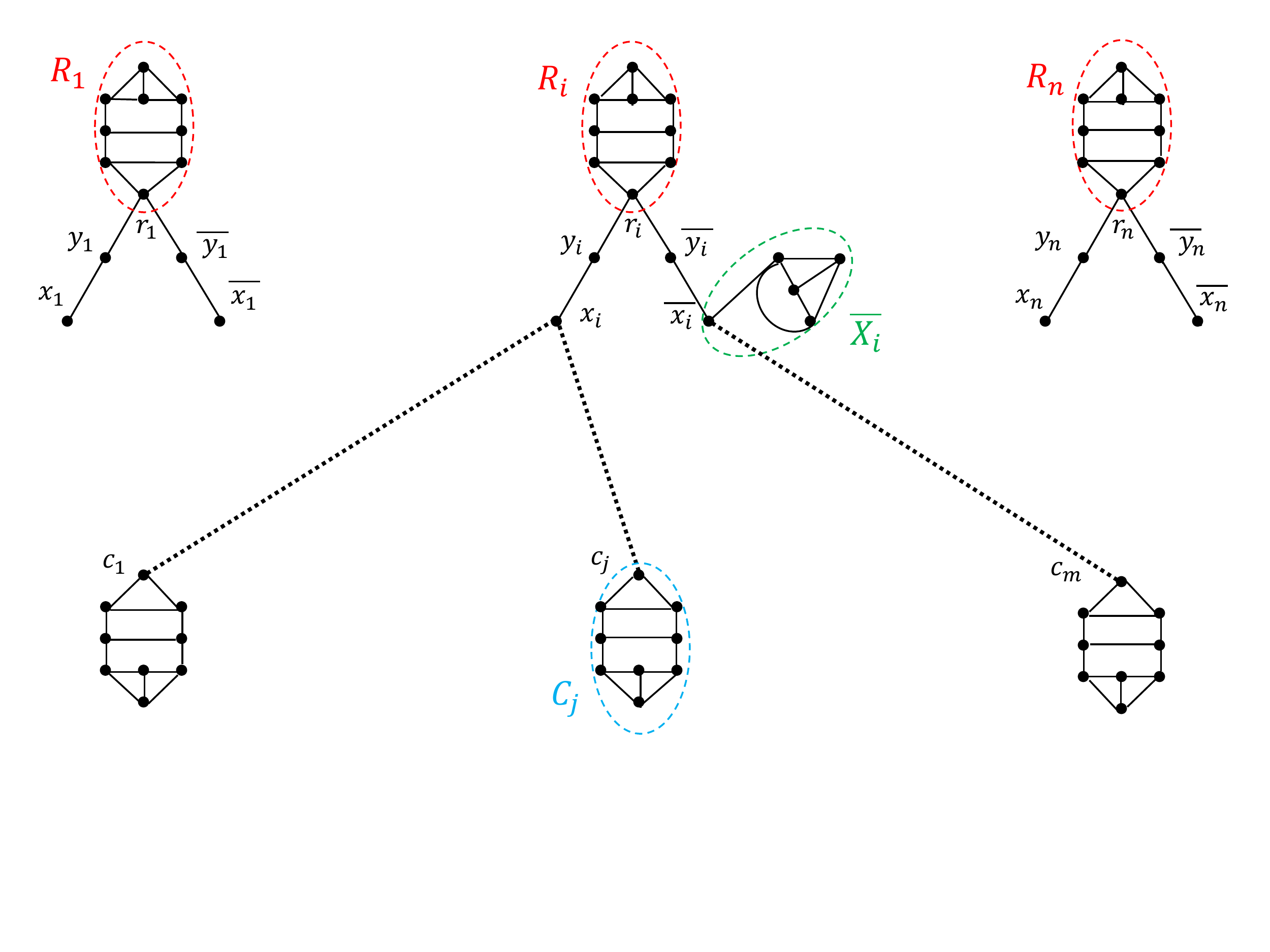}
 \vspace{-30mm}
 \caption{The graph $G$ constructed in Theorem~\ref{thm:planar-k-3}.
 \label{fig:planar-k-3}}
\end{figure}

\begin{theorem}
For $k=3$ the \textsc{Anchored $k$-Core} problem is NP-hard even on planar graphs of maximum degree $5$.
\label{thm:planar-k-3}
\end{theorem}
\begin{proof}
We reduce from the \textsc{Restricted-Planar-3-SAT} problem. For an instance $\phi$ of \textsc{Restricted-Planar-3-SAT} let
$G_{\phi}$ be the associated planar graph. We define two special graphs $T$ and $W$ (see Figure~\ref{fig:t-and-w}) and build
the graph $G$ as follows (see Figure~\ref{fig:planar-k-3}):
\begin{itemize}
\item For each $1\leq i\leq n$ subdivide the edge $r_{i}x_{i}$ and let the newly introduced vertex be $y_i$
\item For each $1\leq i\leq n$ subdivide the edge $r_{i}\overline{x_i}$ and let the newly introduced vertex be
    $\overline{y_i}$.
\item For each $1\leq i\leq [n$ attach a copy of $T$ by identifying its degree two vertex with vertex $r_i$. Call this
    gadget as $R_i$
\item For each $1\leq i\leq n$ if $x_i$ appears in exactly one clause then attach a copy of $W$ by identifying its degree
    one vertex with vertex $x_i$. Call this gadget as $X_i$
\item For each $1\leq i\leq n$ if $\overline{x_i}$ appears in exactly one clause then attach a copy of $W$ by identifying
    its degree one vertex with vertex $\overline{x_i}$. Call this gadget as $\overline{X_i}$
\item For each $1\leq j\leq m$ attach a copy of $T$ by identifying its degree two vertex with vertex $c_j$. Call this
    gadget as $C_j$
\end{itemize}
This completes the construction of the graph $G=(V,E)$. All the gadgets we added are planar, and it is easy to verify that the
planarity is preserved when we construct $G$ from $G_{\phi}$. Let $k=3,\ b=n$ and $p=|V|-2n$. The claim is that $\phi$ is
satisfiable if and only if the instance $(G,b,k,p)$ of $p$-AKC answers YES.

Suppose $\phi$ is satisfiable. For each $1\leq i\leq n$, if $x_i = 1$ in the satisfying assignment for $\phi$ then pick $y_i$
in $B$ and $\overline{x_i},\ \overline{y_i}$ in $B'$. Otherwise pick $\overline{y_i}$ in $B$ and $x_i,\ y_i$ in $B'$. Clearly
$|B|=n$ and $|B'|=2n$. Let $B$ be the set of anchors and set $H=V\setminus B'$. Now the claim is that every vertex $w\in
H\setminus B$ has degree at least three in the induced subgraph $G[H]$. For each $1\leq i\leq n$, exactly one of $y_i$ or
$\overline{y_i}$ is in $B$. Hence $r_i$ (and also each vertex of $R_i$) has degree exactly three in $H$. Consider a literal
$x_i$. We have the following two cases:
\begin{itemize}
\item \underline{$y_i\in B$}: $x_i$ gets one edge from $y_i$. If $x_i$ appears in exactly one clause then it gets one edge
    from that clause vertex and one edge from its neighbor in $X_i$ (and each vertex in $X_i$ has degree at least three in
    $H$). Otherwise $x_i$ gets two edges from the two clause vertices which it appears in.
\item \underline{$\overline{y_i}\in B$}: $\overline{x_i}$ gets one edge from $\overline{y_i}$. If $\overline{x_i}$ appears
    in exactly one clause then it gets one edge from that clause vertex and one edge from its neighbor in $\overline{X_i}$
    (and each vertex in $\overline{X_i}$ has degree at least three in $H$). Otherwise $\overline{x_i}$ gets two edges from
    the two clause vertices which it appears in.
\end{itemize}
Finally consider a clause vertex $c_j$. It has at least one true literal say $x_i$ in it. In addition $c_j$ has two neighbors
in $C_j$, and hence the degree of $c_j$ is at least three in $H$. Consequently each vertex in $C_j$ has degree at least three
in $H$. Therefore with $b=|B|=n$ anchors we can cover a $3$-core of size at least $|V\setminus B'|=|V|-2n=p$, and hence
$(G,b,k,p)$ answers YES.

Suppose that the instance $(G,b,k,p)$ of $p$-AKC answers YES. Let us denote the $3$-core by $H$. Note that we can afford to
not have at most $2n$ vertices in the $3$-core. Each $y_i$ and $\overline{y_i}$ have degree two in $G$: so either we cannot
have them in the $3$-core or we need to pick them as anchors. Also for $i\in [n]$ if we do not pick at least one of $y_i$ or
$\overline{y_i}$ then the vertex $r_i$ also cannot be in the $3$-core. This will lead to a cascade effect and the whole gadget
$R_i$ cannot be the in the $3$-core, which is a contradiction since it has $2n+3$ vertices and we could have left out at most
$n$ vertices from the core. If for some $i\in [n]$ we pick both $y_i$ and $\overline{y_i}$ as anchors then for some $j\neq i$
we cannot pick either of $y_j$ and $\overline{y_j}$ as anchors since the total budget for anchors is at most $n$. Therefore we
must anchor exactly one of $y_i, \overline{y_i}$ for each $1\leq i\leq n$. Let $B$ be the set of anchors. Consider the
assignment $f:\{1,2,\ldots,n\}\rightarrow \{0,1\}$ given by $f(x_i)=1$ if $y_i\in B$ or $f(x_i)=0$ otherwise. We claim that
$f$ is indeed a satisfying assignment for $\phi$. Consider a clause vertex $c_j$. We know that $c_j$ must lie in the $3$-core:
otherwise we lose the entire gadget $C_j$ which has $2n+3$ vertices which is more than our budget. Therefore $c_j$ has an edge
in $G[H]$ to some vertex say $x_i$. If $y_i\notin B$ then $x_i$ can have degree at most two in $G[H]$: either it appears in
exactly one clause and has a copy of $W$ attached to it, or it is adjacent to two clause vertices. Therefore $y_i\in B$ which
implies $f(x_i)=1$, and so the clause $c_j$ is satisfied.

Finally note that the maximum degree of $G$ is five, which can occur if $c_j$ has three literals.
\end{proof}

\subsection{NP-hardness on Planar Graphs for $k\geq 4$}

In this section we show that the \textsc{Anchored $k$-Core} problem is NP-hard on planar graphs for all $k\geq 4$, even in
graphs of maximum degree $k+2$.

\begin{figure}[h]
 \centering
 \includegraphics[height=4.5in]{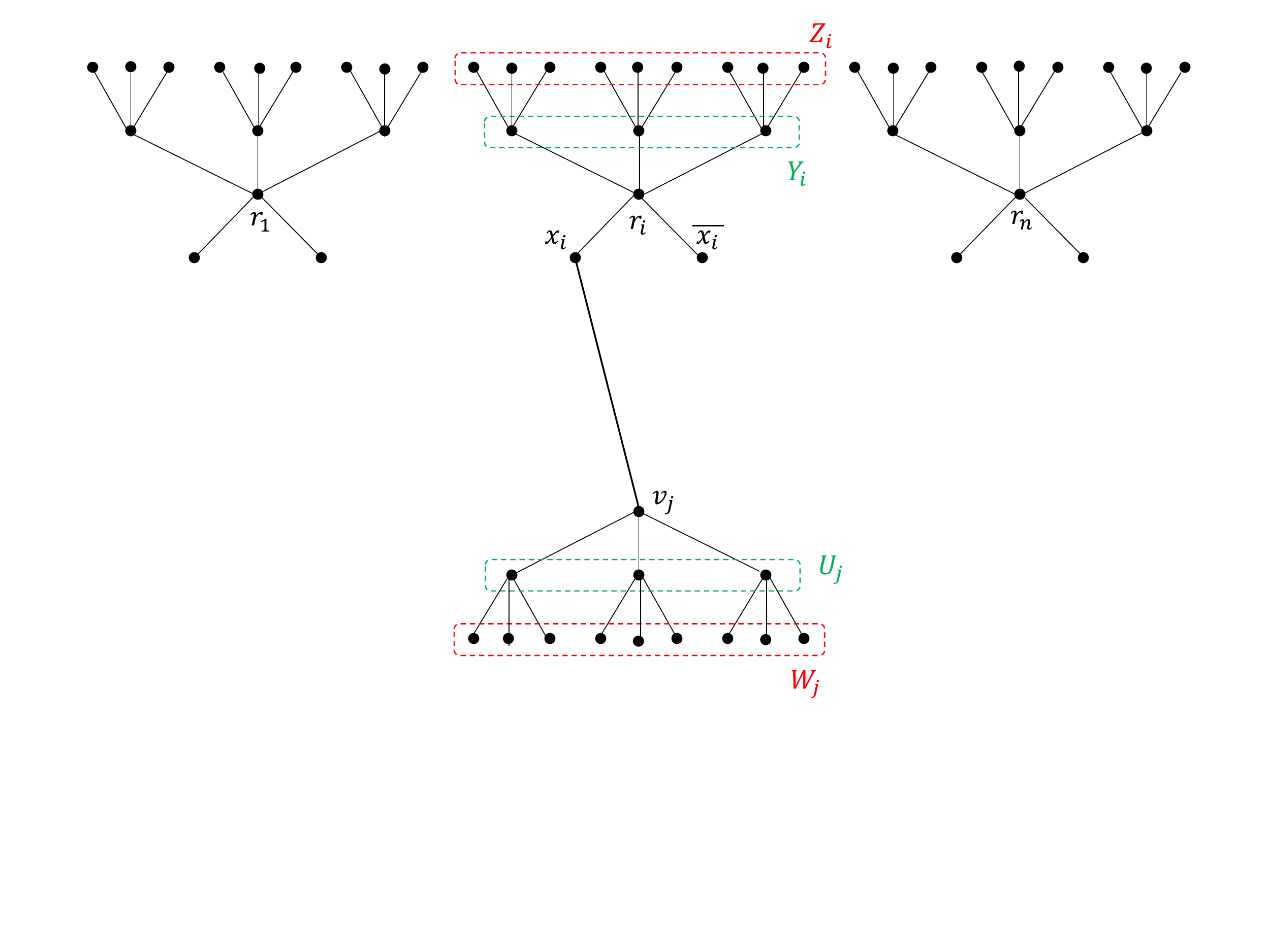}
 \vspace{-30mm}
 \caption{The graph $G$ constructed in Theorem~\ref{thm:planar-k-general} for $k=4$.
 \label{fig:planar-k-general}}
\end{figure}

\begin{theorem}
For any $k\geq 4$ the \textsc{Anchored $k$-Core} problem is NP-hard even on planar graphs of maximum degree $k+2$.
\label{thm:planar-k-general}
\end{theorem}
\begin{proof}
Fix any $k\geq 4$. We reduce from the \textsc{Restricted-Planar-3-SAT} problem. For an instance $\phi$ of
\textsc{Restricted-Planar-3-SAT} let $G_{\phi}$ be the associated planar graph. We build a graph $G=(V,E)$ from $G_{\phi}$ as
follows (see Figure~\ref{fig:planar-k-general}):
\begin{enumerate}
\item For each $1\leq i\leq n$
    \begin{itemize}
    \item Add a set $Y_i$ of $k-1$ vertices and make all of them adjacent to $r_i$.
    \item For each vertex $y\in Y_i$ add $k-1$ vertices and make all of them adjacent to $y$. Let $Z_i$ be the set of
        all these $(k-1)^2$ vertices.
    \end{itemize}
\item For each $1\leq j\leq m$
    \begin{itemize}
    \item Add a set $U_j$ of $k-1$ vertices and make all of them adjacent to $v_j$.
    \item For each vertex $u\in U_j$ add $k-1$ vertices and make all of them adjacent to $u$. Let $W_j$ be the set of
        all these $(k-1)^2$ vertices.
    \end{itemize}
\end{enumerate}
Set $b=n((k-1)^2+1)+m(k-1)^2$ and $p=n(k(k-1)+2)+m(k(k-1)+1)= b+nk+mk$. Note that degree of each $x_i$ and $\overline{x_i}$ is
at most three in $G$. We claim that $\phi$ is satisfiable if and only if the instance $(G,b,k,p)$ of $p$-AKC answers YES.

Suppose $\phi$ is satisfiable. For each $1\leq i\leq n$, if $x_i=1$ in the satisfying assignment then select $x_i$ in $B'$ and
$\overline{x_i}$ in $B''$. Else select $\overline{x_i}$ in $B'$ and $x_i$ in $B''$. Let $B=B'\bigcup (\cup_{i=1}^{n} Z_i)
\bigcup (\cup_{j=1}^{m} W_j)$. Let us place the anchors at vertices of $B$. Then $|B|=n+n(k-1)^2+m(k-1)^2=b$. Now the claim is
that $V\setminus B''$ forms a $k$-core. This would conclude the proof since $|V\setminus B''|=|G|-n=p$. For each $1\leq i\leq
n$ the vertex $r_i$ gets one neighbor from either $x_i$ or $\overline{x_i}$ and it has $k-1$ neighbors in $Y_i$. Each vertex
in $Y_i$ has one neighbor in $r_i$ and $k-1$ neighbors in $Z_i$. Each vertex in $U_j$ has one neighbor in $v_j$ and $k-1$
neighbors in $W_j$.
Since there is a satisfying assignment we know that $v_j$ has at least one neighbor in $B'$, and of course has $k-1$ neighbors
in $U_j$. So $V\setminus B''$ forms a $k$-core with $B$ as the anchor set, and the instance $(G,b,k,p)$ of $p$-AKC answers
YES.

Now suppose that the instance $(G,b,k,p)$ of $p$-AKC answers YES. Let us denote the anchors by $B$ and the $k$-core by $H$.
Consider $S=(\cup_{i=1}^{n} \{x_i,\overline{x_i}\})\bigcup (\cup_{i=1}^{n} Z_i) \bigcup (\cup_{j=1}^{m} W_j)$. Note that
$|S|=b+n$. Any vertex in $S$ has degree at most $\max\{k-1,3\}$ in $G$: so if it is present in the $k$-core then it must be an
anchor. Since $p=|V|-n$ at least $|S|-n$ vertices from $S$ must be anchors. Since $|S|-n=b$ these vertices must be the anchor
set say $B$ and the $k$-core is $H=B\cup (V\setminus S)$. Let $z\in Z_i$ for some $1\leq i\leq n$ be adjacent to $y\in Y_i$ in
$G$. If $z\notin B$ then $y$ has at most $k-1$ neighbors in $H$, which contradicts the fact that $Y_i\subseteq H\setminus B$.
Therefore $Z_i\subseteq B$ for every $1\leq i\leq [n]$. Similarly we have $W_j\subseteq B$ for every $1\leq j\leq m$. So now
we can only choose $n$ more anchors from the set $\cup_{i=1}^{n} \{ x_i,\overline{x_i}\}$. Suppose for some $1\leq i\leq n$ we
have both $x_i \notin B$ and $\overline{x_i}\notin B$. Then the vertex $r_i$ has degree at most $k-1$ in $G[H]$, contradicting
the fact that $r_i\in H\setminus B$. Therefore for every $1\leq i\leq n$ at least one of $x_i$ or $\overline{x_i}$ must be in
$B$. As the budget for the anchors is $n$ we know that for every $1\leq i\leq n$ exactly one of $x_i$ or $\overline{x_i}$ is
in $B$. Consider the assignment $f:\{1,2,\ldots,n\}\rightarrow \{0,1\}$ for $\phi$ given by $f(i)=1$ if $x_i\in B$ and 0
otherwise. The claim is that $f$ is a satisfying assignment for $\phi$. Consider any clause $C_j$ of $\phi$. The vertex $v_j$
has exactly $k-1$ neighbors in $U_j$, and hence must have at least one neighbor in $B$ (which appears in the clause $C_j$). If
this neighbor is some $x_i\in B$ then the assignment would set $x_i=1$. Else the neighbor is of the type $\overline{x_i}\in B$
and then our assignment would have set $x_i=0$. Hence $f$ is a satisfying assignment for $\phi$. Finally note that the maximum
degree of $G$ is $k+2$, which can occur if $v_j$ has three literals.
\end{proof}

\section{FPT on Planar Graphs Parameterized by $b$}

In this section we show that the $p$-AKC problem is FPT parameterized by $b$ on planar graphs when $k\geq 7$.

\begin{lemma}
\label{lem:first-order-logic} The problem of checking whether there is an anchored $k$-core such that $q\geq |H|\geq p$ can be
expressed in first-order logic.
\end{lemma}
\begin{proof}
Consider the following formula in first-order logic: \\
$ \phi_{q} = \bigvee_{p \leq i \leq q}\bigl( \exists h_1,h_2, \ldots h_i:\ H_i \wedge \bigvee_{1 \leq j \leq b} \bigl(\exists
b_1, b_2, \ldots b_j\ : B_j \wedge B_{j}H_{i} \wedge \forall y: \bigl(\bigvee_{1 \leq i_1 \leq i}(y = h_{i_1}) \wedge
YB_{j}\bigr) \rightarrow \exists v_1 \ldots v_k:\ V_k \wedge V_{k}H_{i}
\wedge YV_{k}\bigr)\bigr)$ \\
where\\
$H_i = \bigwedge_{1\leq i_1 \neq i_2\leq i} (h_{i_1} \neq h_{i_2})$\\
$B_j = \bigwedge_{1\leq j_1 \neq j_2\leq j} (b_{j_1} \neq b_{j_2})$\\
$B_{j}H_{i} = \bigwedge_{1 \leq j_1 \leq j}(\bigvee_{1 \leq i_1 \leq i}(b_{j_1} = h_{i_1})$\\
$YB_{j} = \bigwedge_{1 \leq j_1 \leq j}(y \neq b_{j_1})$\\
$V_k = \bigwedge_{1\leq k_1 \neq k_2\leq k}(v_{k_1} \neq v_{k_2})$\\
$V_{k}H_{i} = \bigwedge_{1 \leq k_1 \leq k}(\bigvee_{1 \leq i_1 \leq i}(v_{k_1} = h_{i_1}))$\\
$YV_{k} = \bigwedge_{1 \leq k_1 \leq k}(yv_{kg_1} \in E))$

We claim that the formula $\phi_{q}$ correctly expresses the problem of checking whether there is an anchored $k$-core such
that $q\geq |H|\geq p$. The formulae $H_i,\ B_j$ and $B_k$ just check that all the variables in the respective formula are
pairwise distinct. The formula $B_{j}H_i$ checks every anchor is present in the anchored $k$-core $H$. Finally for every $y\in
H\setminus B$ we enforce that that there are at least $k$ elements $v_1,v_2,\ldots,v_k$ which are pairwise distinct, present
in $H$ and adjacent to $y$. It is now easy to see that any solution $H$ such that $q\geq |H|\geq p$ gives a solution to the
formula $\phi_{q}$ and vice versa, i.e., the formula $\phi_{q}$ exactly expresses this problem. Note that the length of
$\phi_{q}$ is poly$(q)$ since $q\geq p\geq b$ and $q\geq k-1$.
\end{proof}

Seese~\cite{seese1996linear} showed that any graph problem expressible in first-order logic can be solved in linear FPT time
on graphs of bounded degree. More formally, let $\cal{X}$ be a graph problem and $\phi_{\cal{X}}$ be a first-order formula for
$\cal{X}$. For a constant $c>0$ consider the graph class $\mathcal{G}_{c} = \{\ G\ |\ \Delta(G)\leq c \}$. Then for every
$G\in \cal{G}$ we can solve $\cal{X}$ in $O(f(|\phi_{\cal{X}}|)\cdot |G|)$ where $f$ is some function. This was later extended
by Dvorak et al.~\cite{dvorak-kral-thomas} to a much richer graph class known as graphs with \emph{bounded expansion}.
We refer
to~\cite{dvorak-kral-thomas,nesetril-mendez} for the exact definitions. However we remark that examples of such graph classes
are graphs of bounded degree, graphs of bounded genus (including planar graphs), graphs that exclude a fixed (topological)
minor, etc. Using these results we can give FPT algorithm parameterized by $b$ on
some classes of sparse graphs when $k$ is sufficiently large. However for the sake of simplicity we just state the result for
planar graphs.

\begin{lemma}
\label{lem:planar-bounded} Let $G$ be a planar graph on $n$ vertices. Let $k\geq 7$ and $m$ be the set of vertices of degree
at least $k$. Then $\frac{m}{n}<\frac{6}{7}$.
\end{lemma}
\begin{proof}
Summing up the degrees of the graph gives twice the number of edges. For $v\in G$ we define $\text{deg}'(v)=7$ if
$\text{deg}(v)\geq7$, and $\text{deg}'(v)=1$ otherwise. Note that every $v\in $G satisfies $\text{deg}'(v)\leq \text{deg}(v)$.
It is a well known fact that the maximum number of edges in a $n$ vertex planar graph is $3n-6$. Therefore we have $2(3n-6) =
6n-12 \geq 2|E(G)| = \sum_{v\in G} \text{deg}(v) \geq \sum_{v\in G} \text{deg}'(v) = 7m+(n-m)$. Rearranging we get $6m\leq
5n-12$ which implies $\frac{m}{n}\leq \frac{5n-12}{6n} <\frac{6}{7}$.
\end{proof}

%

Now we show that for $k\geq 7$ the $p$-AKC problem is FPT parameterized by $b$ on planar graphs.

\begin{theorem}
Let $k\geq 7$. Then for the class of planar graphs the $p$-AKC problem can be solved in linear FPT time parameterized by the
number of anchors $b$.
\end{theorem}
\begin{proof}
Fix any $k\geq 7$. Suppose there is a solution $H$ for the $p$-AKC problem and let $B$ be the set of anchors. Each vertex in
$H\setminus B$ has degree at least $k$. Applying Lemma~\ref{lem:planar-bounded} to the planar graph $G[H]$ we have
$\frac{|H|-|B|}{|H|} = \frac{|H\setminus B|}{|H|}< \frac{6}{7}$. Therefore $\frac{|B|}{|H|}\geq \frac{1}{7}$ and so $7b\geq
7|B|\geq |H|$. So we can express the problem as checking whether there is an anchored $k$-core $H$ such that $p\leq |H|\leq
7b$. As shown in Lemma~\ref{lem:first-order-logic} we write the first-order formula $\phi_{7b}$ for this problem. By the
result of Dvorak et al.~\cite{dvorak-kral-thomas} the $p$-AKC problem can be solved in $O(f(7b)\cdot n)$ for some function
$f$, i.e., it can be solved in linear FPT time parameterized by $b$.
\end{proof}


\section{Conclusions and Open Problems}
We studied the complexity of the AKC problem on the class of planar graphs, thus answering the question raised
in~\cite{bhawalkar-icalp}. We showed that the AKC problem is NP-hard on planar graphs, even if the graph has maximum degree
$k+2$. We also improve some fixed-parameter intractability results for the $p$-AKC problem. Finally on the positive side we
show that for all $k\geq 7$ the $p$-AKC problem on planar graphs is FPT parameterized by $b$.

There are still several interesting questions remaining. We mention some of them here: what is the parameterized complexity
status of the problem parameterized by $b$ on planar graphs for $3\leq k\leq 6$?
What happens when we consider the problem on random graphs? Can we get reasonable approximation algorithms on some restricted
graph classes?

\medskip
\textbf{Acknowledgements:} We would like to thank the anonymous referees of AAAI 2013 for helpful comments.

%

%


\bibliography{docsdb}
\bibliographystyle{abbrv}

\end{document}